\documentclass[lettersize,journal]{IEEEtran}
\usepackage{amsmath,amsfonts,amssymb,amsthm}
\usepackage{algorithmic}
\usepackage{array}
\usepackage[caption=false,font=normalsize,labelfont=sf,textfont=sf]{subfig}
\usepackage{textcomp}
\usepackage{stfloats}
\usepackage{url}
\usepackage{graphicx}
\usepackage{booktabs}
\usepackage{tabularx}
\usepackage{adjustbox}
\usepackage{verbatim}
\usepackage{cite}
\usepackage[hidelinks]{hyperref}

\def\BibTeX{{\rm B\kern-.05em{\sc i\kern-.025em b}\kern-.08em
		T\kern-.1667em\lower.7ex\hbox{E}\kern-.125emX}}
\graphicspath{{./}}
\newcommand{\TACQE}{\textsc{TACQE}}

\newtheorem{theorem}{Theorem}

\newcolumntype{Y}{>{\raggedright\arraybackslash}X}

\begin{document}
	
	\title{Topology-adaptive Estimation of Maximum Voltage Unbalance for Distribution Networks with Finite-sample Guarantee}
	\author{Jinbao Wang, Jun Liu, Shiliang Zhang \IEEEmembership{Member, IEEE}, Xuehui Ma, He Zhang\thanks{J. Wang, J. Liu, and X. Ma (jbwang@stu.xaut.edu.cn, liujun0310@xaut.edu.cn, xuehui@xaut.edu.cn) are with Xi'an University of Technology, China. S. Zhang (shilianz@ifi.uio.no) is with University of Oslo, Oslo, Norway. H. Zhang (hezhang@xpu.edu.cn) is with Xi’an Polytechnic University, China.} }
	\maketitle
	
		\begin{abstract}
		High penetration of single-phase distributed power production and frequent topology changes can significantly increase voltage unbalance in distribution networks. Reliable probabilistic power flow (PPF) assessment for the network, therefore, must account for topology-dependent power-flow equations, numerical error in the voltage unbalance factor (VUF), and finite-sample uncertainty in power measurements. In this paper, we propose Topology-Adaptive Certified Quantile Estimation (TACQE) to estimate the maximum VUF across monitored electric buses. Following a local network change, TACQE updates the cross-couplings among power-flow sequences, and certifies candidate PPF solutions against the updated nonlinear equations. The resulted power-flow state enclosures provide rigorous bounds on the maximum VUF. Furthermore, we derive independent reference-solver evaluations to calibrate discrepancies from the reference PPF solution. This leads to a uniform sampling bound that yields simultaneous upper-tail probabilities and quantiles, which provides a finite-sample guarantee for the estimated maximum-VUF distribution over an entire range of VUF thresholds. From our numerical simulations across 101 fixed-topology quantiles, TACQE achieves mean absolute VUF-quantile errors of $2.26\times10^{-5}$ and $5.06\times10^{-5}$ percentage points under nominal and stress conditions, respectively. In our topology-change experiments on three distributions networks, its quantile-function and upper-tail errors for the estimated VUF distribution are smaller than those of four baselines by factors of 8.27--12.1 and 7.76--10.4, respectively. With equal reference-solver budgets, TACQE produces a median probability-envelope width over the selected upper-tail range that is approximately 55\% smaller than that of reference Monte Carlo simulations. This indicates much tighter uncertainty bounds for high-VUF exceedance probabilities. Moreover, the local admittance-update routine achieves a median speedup of 12.3 under our approach, demonstrating that the designed TACQE can provide a fast candidate post-change power flow solution where it reuses most of the previous network model and modifies only the affected bus blocks.
	\end{abstract}
	
	\begin{IEEEkeywords}
		Voltage unbalance factor, high-quantile estimation, topology-varying distribution networks, finite-sample guarantees, probabilistic power flow.
	\end{IEEEkeywords}
	
	\section*{Nomenclature}
	{\footnotesize
		\begin{IEEEdescription}[
			\IEEEusemathlabelsep
			\IEEEsetlabelwidth{$\widehat L_{\mathrm{ref}},\widehat U_{\mathrm{ref}}$}
			]
			\item[$X=(\Theta,S)$] Random topology and complex power-injection scenario.
			\item[$M^\star,p_\tau^\star$] Reference maximum VUF over monitored buses and its exceedance probability.
			\item[$Q_M^\star(q)$] Population quantile of the reference maximum VUF.
			\item[$Y_\theta^{abc},\mathbf Y_\theta^{012}$] Phase-domain and fully coupled sequence-domain admittance matrices.
			\item[$\mathcal Y_\theta,Z_\theta$] Reduced admittance matrix and its inverse for topology $\theta$.
			\item[$F_\theta,\widehat x,B_\theta$] Exact power-flow residual, candidate state, and numerical preconditioner.
			\item[$E_\theta,\eta,K_\theta$] Verified inverse-defect bound, preconditioned-residual bound, and comparison matrix.
			\item[$\rho,\mathcal D_\rho$] Positive radius vector and candidate-centered complex polydisc.
			\item[$L^{\mathrm{det}},U^{\mathrm{det}}$] Deterministic lower and upper bounds on maximum VUF at the certified equilibrium.
			\item[$S(X),\ell_k$] Reference-containment expansion score and calibrated order-statistic margin.
			\item[$\alpha_{\mathrm{cal}},\beta$] Population miscoverage bound and calibration failure probability.
			\item[$L^{\mathrm{cal}},U^{\mathrm{cal}}$] Lower and upper endpoints after reference calibration.
			\item[$\widehat p_N^-,\widehat p_N^+$] Empirical lower- and upper-endpoint exceedance rates.
			\item[$\widehat L_{\mathrm{ref}},\widehat U_{\mathrm{ref}}$] Simultaneous lower and upper reference tail-probability bounds.
		\end{IEEEdescription}
	}
	
	\section{Introduction}
	\label{sec:introduction}
	
	Distributed energy resources, electric vehicles, and heat pumps often connect to individual phases of distribution networks. Their time-varying power injections increase phase imbalance, while switching operations alter how these imbalances propagate through the network. Consequently, voltage unbalance may increase substantially. Persistent voltage unbalance raises network losses, stresses power equipment, and reduces operating margins~\cite{Damianakis2025GridImpact,Tsobze2025VoltageImbalance}. The voltage unbalance factor (VUF) is defined as the ratio of negative- to positive-sequence voltage magnitudes. Its maximum across monitored buses characterizes the most severe unbalance in a given operating scenario. Since average operating conditions can mask infrequent high-VUF events, operational assessment requires reliable exceedance probabilities and high quantiles.
	
	Probabilistic VUF assessment typically combines repeated three-phase power-flow calculations, or their approximations, with empirical-distribution estimation or importance sampling~\cite{Gruosso2024AdvancedVUF,Nguyen2024FP3Phi,Christianen2024ImportanceSampling}. Existing studies often address power-flow computational cost, state-estimation accuracy, and tail-sampling efficiency separately. Under topology variation, however, these error sources interact and couple with each other. A reliable assessment must first represent the power-flow equations associated with the realized topology. It should also propagate state uncertainty through the VUF ratio and maximum operation across buses, while accounting for finite-sample uncertainty in the resulting distributional estimates.
	
	Topology changes modify the nodal admittance matrix and the associated phase and sequence couplings in the power-flow equations. Approaches with linear and convex approximations can reduce computational cost by simplifying topology-dependent relations~\cite{Lin2024LinDistS,Mahmoodi2024UnbalancedPF}, but their corrections may become inconsistent with the governing equations after a switching action. Graph-based models represent connectivity and line attributes~\cite{Gao2024PhysicsGCN,Ghamizi2025PowerFlowMultiNet}, whereas topology-transfer and dynamic-boundary approaches support operation across multiple network configurations~\cite{Yang2024TopologyTransfer,Su2024DynamicBoundary}. Although these methods improve structural representation or transferability, they do not by themselves certify that each predicted state satisfies the updated fully coupled power-flow equations. Moreover, further iterations using a nominal-topology matrix cannot eliminate the resulted model mismatch. Reliable VUF assessment therefore requires local network-element admittance updates followed by verification against the equations of the realized topology.
	
	Accurate voltages estimates under the correct topology do not, by themselves, ensure accurate VUF estimates. The negative-to-positive-sequence voltage ratio can amplify voltage error when the positive-sequence component is small, and small perturbations can change the bus at which the maximum VUF accurs. Improved linearizations and data-assisted corrections reduce average voltage error~\cite{Chung2024TaylorRobustPF,Liu2022DDAL3PF}. Fixed-point surrogates, Gaussian processes, and physics-guided graph neural networks can likewise improve prediction accuracy~\cite{Nguyen2024FP3Phi,Glover2025MultiphaseGP,Yang2024PGGNNPPF,Gao2024PhysicsGCN}. Nevertheless, average losses and physics-based penalty terms do not provide a scenario-specific bound on maximum VUF. Rare high-VUF states may have little influence on such average objectives. A usable guarantee therefore requires a state enclosure, propagation of that enclosure through the ratio and maximum operation, and calibration against the reference three-phase power-flow solution.
	
	Although scenario-specific state enclosures can bound the maximum VUF for individual operating conditions, they do not provide guarantees for the distribution of maximum VUF under random injections. Such distributional guarantees require accounting for two additional uncertainties: The possibility that a computed interval does not contain the specified reference-solver result, and the finite-sample error in estimating tail probabilities and quantiles. Predictive distributions and interval power-flow methods quantify uncertainty under a fitted model or prescribed input set~\cite{Glover2025MultiphaseGP,Gao2024BayesianPOPF,Liao2024ProbInterval,Pareek2025PVE}. However, they do not necessarily bound the population frequency with which their intervals fail to contain the output of a designated reference power-flow solver. Similarly, importance sampling, directed sampling, scenario reduction, and correlation-aware assessment improve the representation of rare or dependent operating states~\cite{Christianen2024ImportanceSampling,Martins2024KMedoids,Tan2025RareRisk,Niu2024VoltageRisk}. Yet if the underlying power-flow approximation has systematic error, improved sampling alone does not ensure accurate VUF threshold exceedance probabilities or high quantiles.
    
    To address these two sources of uncertainty, we in this paper combine nonparametric calibration with uniform finite-sample distributional bounds. Nonparametric order statistics can bound, at a prescribed confidence level, the population probability that a calibrated interval fails to contain a reference solver value~\cite{Wilks1941Tolerance}. The Dvoretzky--Kiefer--Wolfowitz (DKW) inequality uniformly bounds the difference between an empirical distribution function and its population counterpart over all thresholds~\cite{Dvoretzky1956DKW,Massart1990DKW}. Applied directly to approximate power-flow outputs, however, the DKW inequality does not account for their discrepancy from the reference target. We address this limitation through interval containment. That is, whenever a calibrated interval contains the reference value, exceedance of the lower interval endpoint implies exceedance of the upper endpoint. Combining the calibrated miscoverage bound with the uniform sampling bound yields a simultaneous envelope for reference upper-tail probabilities. Inverting this envelope yields simultaneous bounds on the corresponding VUF quantiles, including quantile levels selected after the envelop has been inspected.
	
	To this end, we develop Topology-Adaptive Certified Quantile Estimation (TACQE) for maximum VUF assessment under random power injections and topology variation. TACQE integrates topology-updated power-flow equations, numerical certification for individual scenarios, and reference-based statistical calibration to provide simultaneous guarantee for upper-tail probabilities and quantiles. The main contributions are as follows.
	\begin{itemize}
		\item We develop a low-rank admittance update method that preserves the diagonal sequence blocks and intersequence couplings induced by a local topology change. Certification against the updated nonlinear power-flow equations produces a state enclosure for each scenarios and a corresponding interval for the maximum VUF.
		\item We derive independent reference-solver evaluations to calibrate the residual discrepancy between each certified maximum-VUF interval and the reference-solver output. A prespecified monparametric order statistic determines an interval-expansion margin and a high-confidence bound on the population miscoverage probability.
		\item We combine calibrated interval-endpoint exceedances and a uniform DKW bound to construct a simultaneous envelope for reference upper-tail probability. Inverting this envelope yields simultaneous VUF quantile under the same confidence event.
	\end{itemize}

    The reminder of this paper is as follows. Section~\ref{sec:method1} briefs the problem where we estimate VUF under random injections and shifted topology. Section~\ref{sec:method2} constructs the post-change network model in phase coordinates and computes candidate power-flow states. Section~\ref{sec:method3} presents the numerical certification procedures, the resulted state enclosures, and the associated maximum-VUF intervals. Section~\ref{sec:method4} describes the reference-solver calibration and the construction of simultaneous tail-probability and quantile bounds. We conduct numerical simulations in Section~\ref{sec:experiment} and analyze the results in detail, and we conclude this work in Section~\ref{sec:conclusion}.
	
	\section{Topology-Adaptive Certified Quantile Estimation}
	\label{sec:method}  
	
	\subsection{Problem Formulation and Overall Framework}\label{sec:method1}
	
	Let $X=(\Theta,S)\sim P_X$ denote a scenario comprising a topology and complex power injections. A reference three-phase power-flow solver returns the sequence voltages at the monitored buses $\mathcal V_{\mathrm{mon}}$. We use $n_c$ reference evaluations for calibration and $N$ endpoint evaluations for distribution estimation. We define the reference maximum VUF and its exceedance probability by
	\begin{equation}
		M^\star(X)=\max_{i\in\mathcal V_{\mathrm{mon}}}100\frac{|V_{2,i}^\star(X)|}{|V_{1,i}^\star(X)|},\quad
		p_\tau^\star=\mathbb P_X(M^\star>\tau),
		\label{eq:target}
	\end{equation}
	The reference quantile is $Q_M^\star(q)=\inf\{\tau:p_\tau^\star\le1-q\}$. We seek lower and upper bounds on $p_\tau^\star$ that hold simultaneously for every threshold at a prescribed joint confidence level. Inverting these bounds must cover all $q\in(0,1)$ under the same event, including quantile levels selected after inspection.

	The target $M^\star$ must be finite and measurable. Thus, the reference solution must exist and positive-sequence voltage must be nonzero at every monitored bus. The admittance update further requires a voltage reference in each energized island and a square, nonsingular reduced matrix $Y_\theta$. Note that scenarios violating these assumptions are unresolved.
	
	TACQE constructs the nonlinear equations for each realized topology and converts a candidate state into a deterministic VUF interval. In the following sections we will describe how we treat different scenarios, how we conduct calibration to evaluate the interval and reference value on the same $n_c$ scenarios and fixes an expansion margin. We will also show we handle miscoverage and sampling terms to enlarge the empirical tail curves before quantile inversion, and illustrate how we separates the reference-calibration budget from the sample size used to resolve the tail.
	
	\subsection{Topology-dependent Update of the Fully Coupled Admittance Matrix}\label{sec:method2}
	
	This section assembles network elements in phase coordinates, where terminal connections and coupling are explicit, and transforms the complete admittance matrix to sequence coordinates. A local change is factored on its affected terminals and incorporated through a low-rank inverse update. Candidate states are then evaluated against the updated nonlinear equations.
	
	Let $\mathcal A(\theta)$ be the set of energized network elements under topology $\theta$. Matrix $A_k$ maps global phase voltages to the terminals of element $k$, whose primitive admittance matrix is $Y_k^{abc}$. The assembled phase-domain matrix and its sequence-domain transformation are
	\begin{equation}
		Y_\theta^{abc}=\sum_{k\in\mathcal A(\theta)}A_k^\top Y_k^{abc}A_k,\qquad
		\mathbf Y_\theta^{012}=\mathsf T^{-1}Y_\theta^{abc}\mathsf T.
		\label{eq:assembly}
	\end{equation}
	Here, $A_{012\to abc}$ is the per-bus sequence-to-phase transformation, $I_n$ is the identity for $n$ buses, and $T=A_{012\to abc}\otimes I_n$ uses sequence-major ordering. The sequence-domain matrix has the block structure
	\begin{equation}
		\mathbf Y_\theta^{012}=
		\begin{bmatrix}
			Y_{00,\theta}&Y_{01,\theta}&Y_{02,\theta}\\
			Y_{10,\theta}&Y_{11,\theta}&Y_{12,\theta}\\
			Y_{20,\theta}&Y_{21,\theta}&Y_{22,\theta}
		\end{bmatrix}.
		\label{eq:coupled_blocks}
	\end{equation}
	The off-diagonal blocks represent intersequence coupling arising from phase asymmetry in network elements. They vanish when the sequence networks decouple. Retaining them captures the effect of topology changes on negative-sequence voltage and VUF.
	
	We fix the positive-sequence reference voltage and eliminate its variable and equation to obtain the reduced matrix $Y_\theta$ and inverse $Z_\theta$. A local change is supported on the affected terminal rows and columns and is factored as $\Delta Y_\theta=U_\theta C_\theta W_\theta^H$. Here, $H$ denotes conjugate transpose, and subscript 0 denotes the nominal topology. When the required inverses exist, the Woodbury identity gives
	\begin{equation}
		Z_\theta=Z_0-Z_0U_\theta(C_\theta^{-1}+W_\theta^{\mathrm H}Z_0U_\theta)^{-1}W_\theta^{\mathrm H}Z_0.
		\label{eq:woodbury}
	\end{equation}
	The changed primitive admittances determine the low-rank factors and terminal support $\mathcal S$. Structural identifiers verify that this support contains every changed element, and an off-support residual check detects omissions. Factoring $\Delta Y_\theta[\mathcal S,\mathcal S]$ supplies the update coordinates. Equation (4) replaces a full factorization with a solve determined by the update rank, although forming or applying the global inverse still incurs network-size-dependent cost. Unchanged primitives are reused. All diagonal and coupling blocks are updated before certification.
	
	For scenario $(\theta,s)$, let $x=[V_0^\top,V_{1,\mathcal N_1}^\top,V_2^\top]^\top$, where $\mathcal N_1$ excludes the positive-sequence reference variable. Matrix $P$ inserts the unknown voltages into the complete sequence vector, and $R$ selects their current-balance equations. The updated power-flow equations are
	\begin{equation}
		\begin{aligned}
			F_\theta(x,s)&=\mathcal Y_\theta x+c_\theta-\mathcal I_\theta(x,s)=0,\\
			\mathcal I_\theta(x,s)&=R\mathsf T^{-1}\!\left[-\overline{s^{abc}/v^{abc}}\right],
		\end{aligned}
		\label{eq:exact}
	\end{equation}
	where $\bar v_{\mathrm{ref}}^{012}$ contains the prescribed reference-voltage contribution, $v^{abc}=T(Px+\bar v_{\mathrm{ref}}^{012})$, and $c_\theta$ is the corresponding known term in the reduced equations. Division and conjugation in the current expression are componentwise.
	
	\subsection{Scenario-Specific State and VUF Certification}\label{sec:method3}
	
	Given a candidate $\hat x$ and a numerical approximation $B_\theta$ to $Y_\theta^{-1}$, we bound the preconditioner defect, candidate residual, and nonlinear variation within a neighborhood. These bounds produce a state enclosure and, subsequently, an interval for maximum VUF. All quantities use the realized topology. This distinguishes numerical convergence from model consistency: A small residual under the nominal topology need not imply a small residual after switching.
	
	Let $\operatorname{fl}(\cdot)$ denote floating-point evaluation and $\operatorname{up}(\cdot)$ an outward-rounded componentwise upper bound. Let $u$ be unit roundoff and $n_x$ the dimension of $x$. Under the standard floating-point model, define $\gamma_r^{\mathrm R}=ru/(1-ru)$ and $\gamma_m^{\mathrm C}=\gamma_{8m+2}^{\mathrm R}$ for complex inner products of length $m$. With $\hat P_\theta=\operatorname{fl}(B_\theta Y_\theta)$ and $\hat R_\theta=\operatorname{fl}(I-\hat P_\theta)$, the inverse-defect bound is
	\begin{equation}
		\begin{split}
			E_\theta=\operatorname{up}\bigl[&(1+\gamma_1^{\mathrm R})|\widehat R_\theta|\\
			&+\gamma_{n_x}^{\mathrm C}|B_\theta||\mathcal Y_\theta|\bigr]
			\succeq|I-B_\theta\mathcal Y_\theta|.
		\end{split}
		\label{eq:inverse_defect}
	\end{equation}
	Let $\hat f_\theta=\operatorname{fl}(F_\theta(\hat x,s))$, and let $\varepsilon_F$ bound its componentwise error relative to the exact residual. A bound on the preconditioned residual is
	\begin{equation}
		\eta=\operatorname{up}\!\left[
		|\operatorname{fl}(B_\theta\widehat f_\theta)|
		+\gamma_{n_x}^{\mathrm C}|B_\theta||\widehat f_\theta|
		+|B_\theta|\varepsilon_F\right].
		\label{eq:residual_bound}
	\end{equation}
	By construction, $\eta\succeq|B_\theta F_\theta(\hat x,s)|$. Thus, $E_\theta$ bounds the inverse defect, while $\eta$ bounds the preconditioned residual. Certification depends on these verified bounds, not on the training loss or stopping criterion of the candidate solver.
	
	A positive radius vector $\rho$ defines the complex polydisc $D_\rho=\{x:|x-\hat x|\preceq\rho\}$. Absolute values and inequalities are componentwise, and $(a)_+=\max(a,0)$. Define $H_v=TP$ and $H_i=RT^{-1}$. Lower bounds on phase-voltage magnitudes and upper bounds on constant-power current variation are
	\begin{equation}
		\underline v_p(\rho)=\bigl(|\widehat v_p^{abc}|-[|H_v|\rho]_p\bigr)_+,\qquad
		d_p(\rho,s)=\frac{|s_p|}{\underline v_p(\rho)^2}.
		\label{eq:phase_slope}
	\end{equation}
	If $\underline v_p(\rho)>0$ at each nonzero-power terminal, define $G_\theta$ by~\eqref{eq:nonlinear_bound}. At a zero-power terminal, set $d_p=0$. For any $x,y\in D_\rho$, the constant-power current difference satisfies~\eqref{eq:current_lipschitz}. The transformation and terminal-selection matrices then yield the componentwise Lipschitz bound in~\eqref{eq:nonlinear_bound}.
	\begin{equation}
		G_\theta(\rho,s)=\operatorname{up}\!\left(|H_i|\operatorname{diag}(d_p)|H_v|\right)
		\label{eq:nonlinear_bound}
	\end{equation}
	componentwise bounds the growth of the constant-power current map over the box. For any $x,y\in\mathcal D_\rho$, the constant-power map satisfies
	\begin{equation}
		\begin{aligned}
			|I_p(x,s)-I_p(y,s)|
			&=|s_p|\frac{|V_p(x)-V_p(y)|}{|V_p(x)||V_p(y)|}\\
			&\leq d_p(\rho,s)|V_p(x)-V_p(y)|.
		\end{aligned}
		\label{eq:current_lipschitz}
	\end{equation}
	Combining nonlinear variation with the inverse defect gives the nonnegative comparison matrix
	\begin{equation}
		K_\theta(\rho,s)=E_\theta+|B_\theta|G_\theta(\rho,s).
		\label{eq:comparison}
	\end{equation}
	Certification requires positive phase-voltage lower bounds at nonzero-power terminals and positive positive-sequence lower bounds at monitored buses. It also requires the componentwise inequality
	\begin{equation}
		\eta+K_\theta(\rho,s)\rho\prec\rho.
		\label{eq:certificate}
	\end{equation}
	Condition~\eqref{eq:certificate} makes the bound on residual and local variation strictly smaller than the radius. Together with voltage positivity, it ensures that the preconditioned map preserves $D_\rho$ and contracts in a weighted infinity norm. The resulting certificate bounds the state error.
	
	These computable conditions establish a unique equilibrium and the componentwise state enclosure needed for VUF propagation.
	
	\begin{theorem}[A posteriori state certification]
    Suppose that $Y_\theta$ is nonsingular and that $E_\theta$, $\eta$, and $G_\theta$ are valid componentwise upper bounds as defined above. If a strictly positive vector $\rho$ satisfies~\eqref{eq:certificate} and both voltage-positivity conditions, then~\eqref{eq:exact} has a unique solution $x_\theta^{\mathrm{eq}}$ in $D_\rho$, with $|x_\theta^{\mathrm{eq}}-\hat x|\preceq\rho$.
	\end{theorem}
	
	\begin{proof}
    Let $\mathcal T_\theta(x)=x-B_\theta F_\theta(x,s)$. The verified bounds give $|\mathcal T_\theta(x)-\hat x|\preceq\eta+K_\theta\rho\prec\rho$, so $\mathcal T_\theta$ maps $D_\rho$ into itself. In the norm $\|z\|_{\rho,\infty}=\max_j|z_j|/\rho_j$, its contraction factor is at most $\max_j[K_\theta\rho]_j/\rho_j<1$. Banach's theorem gives a unique fixed point in $D_\rho$ and the stated enclosure. Moreover, $\|I-B_\theta Y_\theta\|_{\rho,\infty}<1$, so $B_\theta Y_\theta$, and hence $B_\theta$, is nonsingular. A fixed point therefore solves~\eqref{eq:exact}.
	\end{proof}

    A numerical floor $\varepsilon>0$ and inflation factor $\kappa>1$ initialize $\rho^{(0)}=\operatorname{up}[\kappa\max(\eta,\varepsilon)+\varepsilon]$. The search enlarges unsuccessful radii using the right-hand side of~\eqref{eq:certificate}. Diagonal scaling may propose tighter enclosures. Acceptance always requires~\eqref{eq:certificate} with a verified positive margin. Fixed-topology tests use one Newton step and one bounded radius search. The Newton step improves the candidate but does not certify it. Search exhaustion, failed voltage positivity, or numerical failure triggers a reference-solver fallback.

	For sequence component $m$ at monitored bus $i$, define $\underline v_{m,i}=(|\hat V_{m,i}|-\rho_{m,i})_+$ and $\overline v_{m,i}=|\hat V_{m,i}|+\rho_{m,i}$. If every positive-sequence lower bound is positive, the ratio and maximum are bounded by
	\begin{equation}
		L^{\mathrm{det}}=\max_i100\frac{\underline v_{2,i}}{\overline v_{1,i}},\qquad
		U^{\mathrm{det}}=\max_i100\frac{\overline v_{2,i}}{\underline v_{1,i}}.
		\label{eq:vuf_interval}
	\end{equation}
	Thus, $L^{\mathrm{det}}\le M^{\mathrm{eq}}\le U^{\mathrm{det}}$ for the certified equilibrium. We define Tier I where we use these endpoints, and we define Tier II to use the reference interval $[M^\star,M^\star]$ after certification fails. The point estimate is maximum VUF at the accepted center, or $M^\star$ after fallback. It is not the interval midpoint. Fallback solutions add to the reference budget. If the reference solver also fails, the scenario remains unresolved and cannot be silently discarded.
    
	\subsection{Reference Calibration and Tail-probability Bounds}\label{sec:method4}
	
	The deterministic interval encloses the equilibrium of~\eqref{eq:exact}, whereas $M^\star$ is the output of the reference solver. These targets need not coincide. Independent reference evaluations calibrate the discrepancy. Interval containment then relates endpoint and reference exceedances, while DKW controls evaluation-sample error. The two uncertainty sources remain explicit in the resulting population envelope.
	
	For Tier I, define the smallest symmetric expansion that makes the deterministic interval contain the reference value:
	\begin{equation}
		S(X)=\max\{L^{\mathrm{det}}-M^\star,M^\star-U^{\mathrm{det}},0\}.
		\label{eq:score}
	\end{equation}
	Tier II has score zero because its interval already contains $M^\star$. Calibration and evaluation scenarios are independent samples from $P_X$. The interval-producing procedure is fixed independently of the calibration outcomes. For a target miscoverage $\alpha_{\mathrm{cal}}$ and calibration failure probability $\beta$, choose a rank $k\in\{1,\ldots,n_c\}$ before inspecting the scores such that
	\begin{equation}
		\sum_{j=k}^{n_{\mathrm c}}\binom{n_{\mathrm c}}j(1-\alpha_{\mathrm{cal}})^j
		\alpha_{\mathrm{cal}}^{n_{\mathrm c}-j}\leq\beta,
		\label{eq:order}
	\end{equation}
	Set $\ell_k=S_{(k)}$, where $S_{(k)}$ is the $k$th smallest calibration score. The order-statistic bound ensures $P_X(S>\ell_k)\le\alpha_{\mathrm{cal}}$ with confidence at least $1-\beta$. Expand Tier I intervals by $\ell_k$, truncate their lower endpoints at zero, and leave Tier II intervals unchanged, as in~\eqref{eq:calibrated_interval}.
	\begin{equation}
		[L^{\mathrm{cal}},U^{\mathrm{cal}}]=
		\begin{cases}
			[(L^{\mathrm{det}}-\ell_k)_+,U^{\mathrm{det}}+\ell_k],&X\in\mathcal X_{\mathrm I},\\
			[M^\star,M^\star],&X\in\mathcal X_{\mathrm{II}}.
		\end{cases}
		\label{eq:calibrated_interval}
	\end{equation}
	To see why the bound holds, let $F_S$ be the score distribution. Failure means $F_S(S_{(k)})<1-\alpha_{\mathrm{cal}}$. At least $k$ scores must then lie in a lower set having probability at most $1-\alpha_{\mathrm{cal}}$. The binomial upper tail in (15) bounds this event, including distributions with atoms. Thus, with calibration confidence at least $1-\beta$, the probability that $M^\star$ falls outside (16) is at most $\alpha_{\mathrm{cal}}$.
	
	On the event that the calibrated interval contains $M^\star$, its endpoints satisfy the following exceedance inclusions for every threshold:
	\begin{equation}
		\{L^{\mathrm{cal}}>\tau\}\subseteq\{M^\star>\tau\}
		\subseteq\{U^{\mathrm{cal}}>\tau\}.
		\label{eq:event_inclusion}
	\end{equation}
	These inclusions hold on the coverage event, not unconditionally. Define the empirical endpoint exceedance rates $\hat p_N^-(\tau)=N^{-1}\sum_{n=1}^N\mathbf1\{L_n^{\mathrm{cal}}>\tau\}$ and $\hat p_N^+(\tau)=N^{-1}\sum_{n=1}^N\mathbf1\{U_n^{\mathrm{cal}}>\tau\}$, where $\mathbf1$ is the indicator function. The next result accounts for both interval miscoverage and evaluation-sample variation.
	
	\begin{theorem}[Uniform tail-probability envelope]
    Let $\epsilon_N(\delta)=\sqrt{\log(2/\delta)/(2N)}$. Suppose calibration and evaluation scenarios are independent and identically distributed under $P_X$, the interval procedure is fixed independently of calibration outcomes, and the prespecified rank satisfies~\eqref{eq:order}. Then~\eqref{eq:envelope} holds simultaneously for all $\tau\in\mathbb R$ with joint confidence at least $(1-\beta)(1-\delta)$.
		\begin{equation}
			\begin{aligned}
				\widehat L_{\mathrm{ref}}(\tau)
				&=\max\{0,\widehat p_N^--\epsilon_N-\alpha_{\mathrm{cal}}\},\\
				\widehat U_{\mathrm{ref}}(\tau)
				&=\min\{1,\widehat p_N^++\epsilon_N+\alpha_{\mathrm{cal}}\},\\
				\widehat L_{\mathrm{ref}}(\tau)&\leq p_\tau^\star\leq\widehat U_{\mathrm{ref}}(\tau)
			\end{aligned}
			\label{eq:envelope}
		\end{equation}
		holds simultaneously for all $\tau\in\mathbb R$.
	\end{theorem}
	
	\begin{proof}
		On calibration success, the miscoverage set has probability at most $\alpha_{\mathrm{cal}}$. Project $M^\star$ onto $[L^{\mathrm{cal}},U^{\mathrm{cal}}]$ to obtain $\widetilde M$. The projection changes $M^\star$ only on this set, so their population tails differ by at most $\alpha_{\mathrm{cal}}$ uniformly in $\tau$. The pointwise ordering $L^{\mathrm{cal}}\le\widetilde M\le U^{\mathrm{cal}}$ also orders their empirical tails. Conditional on calibration, one two-sided DKW bound applies to the independent sample of $\widetilde M$, with probability at least $1-\delta$. Combining the ordering, DKW error, and miscoverage bound gives (18). Multiplying the calibration probability by the conditional evaluation probability gives the stated joint confidence.
	\end{proof}
	
	Invert the envelopes as in~\eqref{eq:quantile}, using $\inf\varnothing=+\infty$ and $M^\star\ge0$.
	\begin{equation}
		\begin{aligned}
			\widehat Q_{\mathrm{ref}}^-(q)&=\inf\{\tau\geq0\mid\widehat L_{\mathrm{ref}}(\tau)\leq1-q\},\\
			\widehat Q_{\mathrm{ref}}^+(q)&=\inf\{\tau\geq0\mid\widehat U_{\mathrm{ref}}(\tau)\leq1-q\}.
		\end{aligned}
		\label{eq:quantile}
	\end{equation}
	The upper-envelope threshold set is contained in the reference set, which is contained in the lower-envelope set. Taking infima therefore gives $\hat Q_{\mathrm{ref}}^-(q)\le Q_M^\star(q)\le\hat Q_{\mathrm{ref}}^+(q)$ simultaneously for all $q\in(0,1)$. This remains valid after selecting a quantile level from the displayed envelope. Simultaneous validity does not ensure a finite upper bound: if $1-q<\alpha_{\mathrm{cal}}+\epsilon_N$, the upper bound in~\eqref{eq:quantile} is infinite.
    
    Note that coverage is marginal under $P_X$. It does not automatically hold conditional on each topology or after a distribution shift. Also note that the budgets have distinct roles: $n_c$ controls calibration uncertainty, while $N$ controls empirical tail resolution.

	\section{Simulation Results and Discussions}
	\label{sec:experiment}

	\subsection{Experimental Design and Metrics}
	In the simulations, we evaluate the IEEE European Low-Voltage (EuLV) feeder~\cite{IEEEPES2015EuLV}, the CIGRE low-voltage benchmark~\cite{CIGRE2014Benchmark}, and the Kerber rural network~\cite{Kerber2011LV}. Under nominal EuLV conditions, the load multiplier and PV penetration follow uniform distributions $U[0.5,1.5]$ and $U[0,0.8]$, respectively. The probability of concentrated connection to phase $a$ is 0.5. Stress conditions use $U[0.5,2.5]$, $U[0,1]$, and 0.7. These changes increase voltage-drop variability, reverse flow, and negative-sequence excitation.

	The baselines are reference Monte Carlo (Reference MC)~\cite{Gruosso2024AdvancedVUF}, principal-component-analysis support-vector regression (PCA-SVR)~\cite{Jolliffe2002PCA,Drucker1997SVR}, PCA Gaussian-process regression (PCA-GP)~\cite{Jolliffe2002PCA,Rasmussen2006GPML,Glover2025MultiphaseGP}, data-driven-aided linear three-phase power flow (DDA-L3PF)~\cite{Liu2022DDAL3PF}, and a fixed-point three-phase surrogate (FP-$3\phi$)~\cite{Nguyen2024FP3Phi}. Each trained baseline uses 800 reference training evaluations. Calibrated methods use a separate 1000-scenario reference set. Reference MC uses 1000 reference solutions to form its empirical distribution. TACQE requires no training and evaluates 9000 endpoint scenarios. Separate benchmarks measure fallback demand and local-update cost.
	
	The primary metric is mean absolute error (MAE) over the evaluated quantile function. Fixed-topology tests use the 101 levels in (20), with $E_{\mathrm{all}}$ defined by (21).
	\begin{equation}
		\mathcal Q_{101}=\{0.01,0.02,\ldots,0.99\}\cup\{0.025,0.975\},
	\end{equation}
	and define
	\begin{equation}
		E_{\mathrm{all}}=\frac{1}{|\mathcal Q|}\sum_{q\in\mathcal Q}|\widehat Q_M(q)-\widehat Q_{M,\mathrm{eval}}^\star(q)|.
		\label{eq:exp_allq}
	\end{equation}
	Here, $\hat Q_{M,\mathrm{eval}}^\star$ is the empirical reference quantile on the same evaluation scenarios. For methods using all evaluation scenarios, this comparison measures computational distortion rather than error relative to the unknown population quantile $Q_M^\star$. Reference MC uses only 1000 reference solutions and also incurs sampling error relative to the larger evaluation set. Population uncertainty is treated separately by~\eqref{eq:quantile}. Define $\mathcal Q_L=\{0.01,0.025,0.05,0.10\}$ and $\mathcal Q_U=\{0.90,0.95,0.975,0.99\}$, and compute $E_L$ and $E_U$ using~\eqref{eq:exp_allq}. Baseline-to-TACQE error ratios $A_L$ and $A_U$ exceed one when TACQE has smaller error. VUF values are in percentages, and absolute differences are measured by percentage points (p.p.).
	
	Topology experiments use five independent variants of each network, each with one to four sequential changes and 1000 test scenarios per realized topology. We evaluate 11 levels: $\mathcal Q_L$, $\{0.25,0.5,0.75\}$, and $\mathcal Q_U$. The normalized error $I_{m,k}=E_{m,k}/E_{m,0}$ compares method $m$ at change level $k$ with its own unedited-topology error. A separate mechanism diagnostic retains 62,875 records for which both the reference solver and the coupled sequence model converge, including unedited cases. Its enclosure coverage is evaluated only among accepted records.
	
	\subsection{Quantile Accuracy under a Fixed Topology}

	The fixed-topology tests isolate state-solution and VUF-propagation errors. The simulation results show that TACQE has the lowest MAE over the 101 evaluated quantiles in both regimes (Fig.~\ref{fig:full_quantile} and Table~\ref{tab:quantile_mae}): $2.26\times10^{-5}$ p.p. under nominal conditions and $5.06\times10^{-5}$ p.p. under stress. These errors are 12.9\% and 1.30\% below FP-$3\phi$, respectively, and approximately 553--$7.33\times10^3$ times smaller than those of the other trained baselines.
	
	\begin{figure}[!t]
		\centering
		\includegraphics[width=.98\columnwidth]{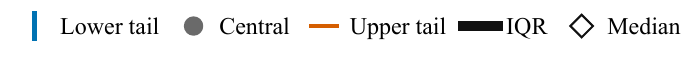}\\[-1mm]
		\includegraphics[width=.485\columnwidth]{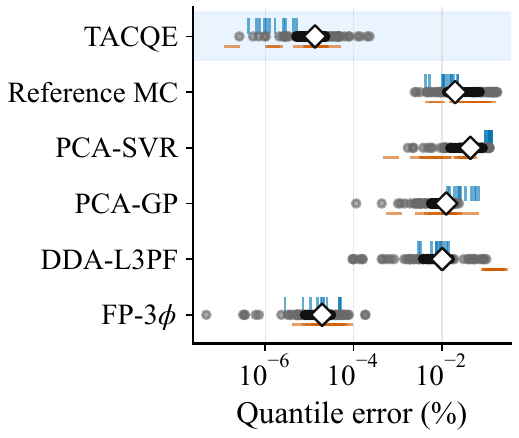}\hfill
		\includegraphics[width=.485\columnwidth]{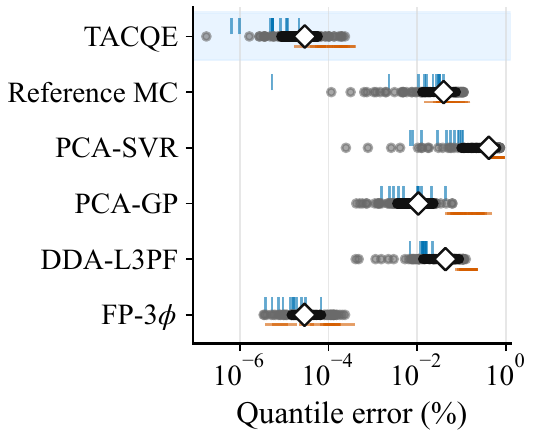}
		\caption{Absolute errors at 101 fixed-topology quantiles under (a) nominal and (b) stress conditions. Points identify quantile groups, diamonds mark medians, and black segments show interquartile ranges (IQRs). Errors are in percentage points.}
		\label{fig:full_quantile}
	\end{figure}
	
	\begin{table}[!t]
		\centering
		\caption{Mean Absolute Error over the Full Quantile Function (\%)}
		\label{tab:quantile_mae}
		\scriptsize
		\setlength{\tabcolsep}{3.2pt}
		\begin{tabular}{lrrr}
			\toprule
			Method & Nominal & Stress & Mean over topology changes\\
			\midrule
			\TACQE{} & $\mathbf{2.26\times10^{-5}}$ & $\mathbf{5.06\times10^{-5}}$ & $\mathbf{1.25\times10^{-2}}$\\
			Reference MC & $4.73\times10^{-2}$ & $4.43\times10^{-2}$ & --\\
			PCA-SVR & $5.29\times10^{-2}$ & $3.71\times10^{-1}$ & $1.51\times10^{-1}$\\
			PCA-GP & $1.35\times10^{-2}$ & $2.80\times10^{-2}$ & $1.43\times10^{-1}$\\
			DDA-L3PF & $3.01\times10^{-2}$ & $5.51\times10^{-2}$ & $1.04\times10^{-1}$\\
			FP-$3\phi$ & $2.60\times10^{-5}$ & $5.13\times10^{-5}$ & $1.25\times10^{-1}$\\
			\bottomrule
		\end{tabular}
	\end{table}
	
	FP-$3\phi$ is the closest point-estimation baseline, with both methods reaching the $10^{-5}$-p.p. error scale. TACQE achieves comparable accuracy without training and also returns scenario-specific intervals. Reference MC estimates the distribution from 1000 reference solutions. TACQE uses 1000 solutions for calibration and 9000 candidate evaluations for distribution estimation. 
	
	Tail performance is less uniform (Fig.\ref{fig:tail_advantage}(a)). Under nominal conditions, TACQE's lower- and upper-tail errors are 0.0800 and 0.395 times those of FP-$3\phi$. Under stress, the lower-tail ratio is 0.434, but the upper-tail errors are $1.36\times10^{-4}$ p.p. for TACQE and $7.10\times10^{-5}$ p.p. for FP-$3\phi$. Thus, lower overall MAE does not establish superiority at every tail quantile.
    
    The VUF ratio and maximum over buses can amplify different state errors and reorder nearby tail samples. TACQE's center estimate is not optimized for individual order statistics. Its certificate instead bounds sequence magnitudes before applying the ratio and maximum, and calibration addresses discrepancies from the reference result.
	
	\begin{figure}[!t]
		\centering
		\includegraphics[width=.485\columnwidth]{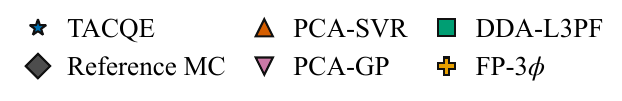}\hfill
		\includegraphics[width=.385\columnwidth]{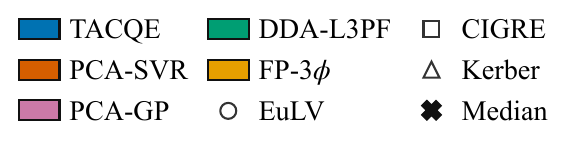}\\[-1mm]
		\includegraphics[width=.485\columnwidth]{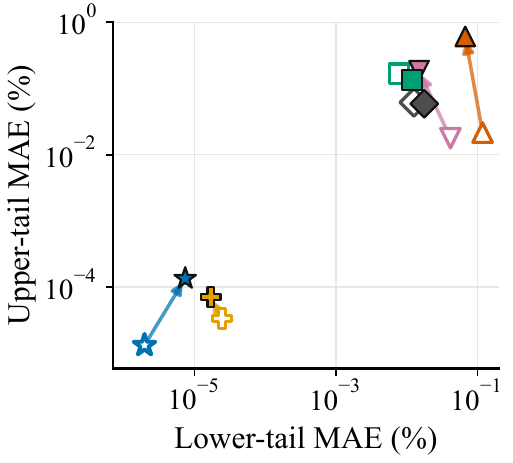}\hfill
		\includegraphics[width=.485\columnwidth]{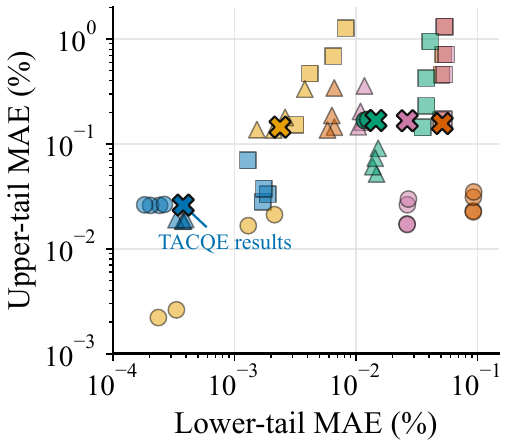}
		\caption{Lower- and upper-tail MAE under (a) fixed topology and (b) topology changes. Arrows in (a) connect nominal and stress results. Crosses in (b) mark medians over 12 network/change combinations. Both axes are logarithmic; errors are in percentage points.}
		\label{fig:tail_advantage}
	\end{figure}
	
	\subsection{Quantile Accuracy under Topology Changes}

    Topology changes separate the methods more clearly (Table I). Across three networks and four change levels, TACQE's mean $E_{\mathrm{all}}$ is 0.0125 p.p., compared with 0.151, 0.143, 0.104, and 0.125 p.p. for PCA-SVR, PCA-GP, DDA-L3PF, and FP-$3\phi$, respectively. The reduction factors from the simulation results are 8.27--12.1. Mean lower-tail errors of the baselines are 4.02--67.1 times larger, and mean upper-tail errors are 7.76--10.4 times larger (Fig.~\ref{fig:tail_advantage}(b)).
	
	In the simulations, only TACQE updates the full realized admittance matrix before certification. PCA-SVR and PCA-GP retain nominal-topology coordinates. DDA-L3PF updates its linear approximation but retains the learned correction. FP-$3\phi$ retains the nominal admittance matrix. Calibration can expand intervals around biased outputs, but it does not repair the point ordering used for empirical quantiles. Some EuLV upper-tail cases still favor FP-$3\phi$ or PCA-GP. On CIGRE and Kerber, TACQE's quantile-function errors are 15.1 and 10.5 times smaller than FP-$3\phi$.

    A small solver residual does not remove structural mismatch. At CIGRE change level three, the fixed-positive-sequence diagnostic gives a median nominal-model residual of $8.80\times10^{-14}$ p.u., but a median topology-induced map change of $9.95\times10^{-3}$ p.u. Further iteration on the nominal equations cannot eliminate this difference. TACQE updates the equations first, then uses~\eqref{eq:vuf_interval} to bound propagation of the remaining state error into maximum VUF.

	Admittance updating and recalibration address different errors. Updating corrects the equations after switching, and recalibration bounds the remaining discrepancy from the reference solver under a specified distribution. Note that recalibration alone does not repair point ordering, and updating alone does not preserve statistical coverage after the injection distribution changes.
	
	\subsection{Coverage and Width of the Tail-probability Envelope}

    For each fixed-topology regime, we use 1000 reference-calibration scenarios and 9000 held-out evaluation scenarios to assess interval violations and envelope width. Expansion scores measure reference discrepancy, while a decomposition over 260 thresholds identifies the width contributions. Equation~\eqref{eq:order} provides the population miscoverage bound, and held-out violation counts provide an implementation check.

    Under stress conditions, the held-out 99th-percentile expansion score is $2.72\times10^{-4}$ p.p. for TACQE, compared with 2.01, 0.460, and 0.369 p.p. for calibrated PCA-SVR, PCA-GP, and DDA-L3PF (Fig.~\ref{fig:mechanism}(a)). FP-$3\phi$ reaches $3.15\times10^{-4}$ p.p. but requires 800 additional training solutions. Larger calibration scores generally require larger endpoint corrections. Probability-envelope width also depends on how many endpoints cross each threshold, so it need not scale proportionally with the expansion margin. TACQE's median widths are 0.0386 nominally and 0.0387 under stress.
	
	\begin{figure}[!t]
		\centering
		\includegraphics[width=.485\columnwidth]{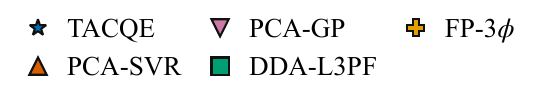}\hfill
		\includegraphics[width=.485\columnwidth]{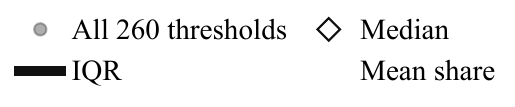}\\[-1mm]
		\includegraphics[width=.485\columnwidth]{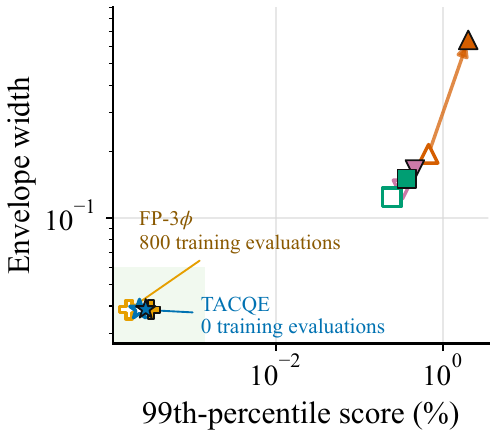}\hfill
		\includegraphics[width=.485\columnwidth]{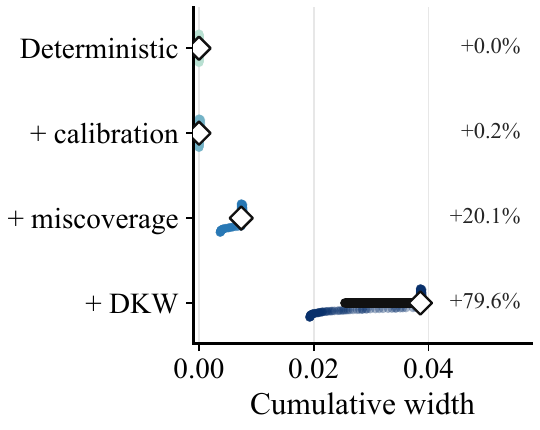}
		\caption{Reference discrepancy and probability-envelope width. (a) Held-out 99th-percentile expansion score versus median width under nominal and stress conditions. (b) Successive width contributions at 260 thresholds. Scores are in percentage points; widths are probabilities.}
		\label{fig:mechanism}
	\end{figure}

    Statistical terms account for almost all envelope width (Fig.~\ref{fig:mechanism}(b)). The deterministic interval and reference-calibration expansion contribute mean incremental shares of 0.0179\% and 0.213\%, respectively, totaling approximately 0.231\%. The miscoverage and DKW terms contribute 20.1\% and 79.6\%. Unlike endpoint corrections, these statistical terms expand both sides uniformly before clipping. Finite-sample uncertainty therefore limits resolution in these tests.

    With $k=n_c=1000$, $\beta=\delta=0.025$, and $N=9000$, the miscoverage bound is $\alpha_{\mathrm{cal}}=0.00368$ and the DKW radius is $\epsilon_N=0.0156$. The joint confidence is at least $0.975^2=0.950625$. Before clipping, these terms contribute the two-sided width in~\eqref{eq:width_substitution}.

	\begin{equation}
		\begin{aligned}
			2(\alpha_{\mathrm{cal}}+\epsilon_N)
			&=2\left[1-0.025^{1/1000}+\sqrt{\frac{\log80}{18000}}\right]\\
			&\simeq0.0386.
		\end{aligned}
		\label{eq:width_substitution}
	\end{equation}
	The resulting width, approximately 0.0386, matches the nominal median and is close to the stress median of 0.0387. The expansion margin nearly doubles, from $3.21\times10^{-4}$ to $6.39\times10^{-4}$ p.p., while median probability width changes by only 0.0001. Endpoint crossings, rather than margin size alone, determine the probability-width increment. Increasing $N$ reduces DKW error as $N^{-1/2}$; increasing $n_c$ permits a smaller miscoverage bound. For these budgets, $\alpha_{\mathrm{cal}}+\epsilon_N\approx0.0193$, so (19) gives an infinite upper bound at $q=0.99$. Accurate empirical 99th-percentile estimates therefore remain distinct from informative population bounds at that level.

	\begin{figure}[!t]
		\centering
		\includegraphics[width=.82\columnwidth]{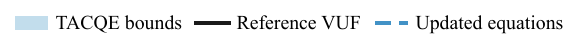}\\[-1mm]
		\includegraphics[width=.485\columnwidth]{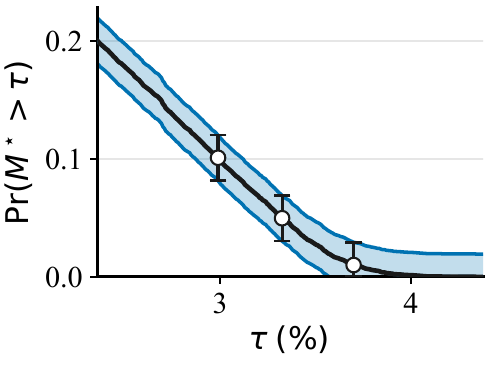}\hfill
		\includegraphics[width=.485\columnwidth]{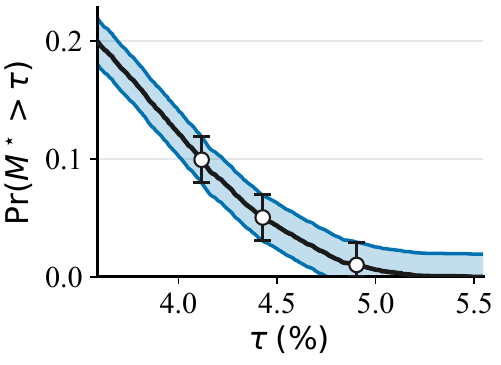}
		\caption{Reference-VUF probability envelopes under (a) nominal and (b) stress conditions. The black curve is the empirical reference tail from 9000 scenarios; blue boundaries are TACQE bounds. The dashed curve is the empirical tail computed from the updated equations.}
		\label{fig:tail_containment}
	\end{figure}
	
	\begin{table}[!t]
		\centering
		\caption{Held-Out Reference Coverage and Envelope Width}
		\label{tab:envelope_validity}
		\scriptsize
		\setlength{\tabcolsep}{2.8pt}
		\begin{tabular}{lrrrrc}
			\toprule
			Distribution & $\ell$/\% & Violations & Rate/\% & $W_{50}$ & $C_{\mathrm{all}}$\\
			\midrule
			Nominal & $3.21\times10^{-4}$ & $6/9000$ & 0.0667 & 0.0386 & Yes\\
			Stress & $6.39\times10^{-4}$ & $2/9000$ & 0.0222 & 0.0387 & Yes\\
			\bottomrule
		\end{tabular}
	\end{table}

    Held-out interval violations are 6/9000 nominally and 2/9000 under stress (Table~\ref{tab:envelope_validity}), corresponding to 0.0667\% and 0.0222\%. The reported upper 95\% Clopper--Pearson limits are 0.145\% and 0.0803\%, both below the calibration bound of 0.368\%. Both envelopes contain the empirical reference tail at all 260 inspected thresholds. Define $C_{\mathrm{all}}$ as the indicator of this containment. Because endpoints and reference values use the same evaluation scenarios, $C_{\mathrm{all}}=1$ checks implementation consistency, and the population guarantee comes from Theorem 2.

	At equal reference budgets, the methods allocate computation differently. Reference MC forms an empirical distribution from 1000 reference solutions. TACQE uses 1000 reference solutions for calibration and 9000 certified endpoint scenarios for tail estimation. In the fixed-topology tests, no fallback occurs, and TACQE's median envelope width is approximately 55\% smaller over thresholds with empirical reference exceedance probabilities from 0.01 to 0.20.

	This reduction depends on successful certification. Every Tier II fallback adds a reference solution and reduces the budget advantage. Note that it does not change the statistical argument, which already includes fallback intervals.
	
	\subsection{Effects of Admittance-matrix Updating on Certification}
	\label{sec:topology_mechanisms}

    A paired diagnostic isolates the admittance update. Both calculations use the same fixed positive-sequence voltage and 12 iterations of the coupled zero- and negative-sequence model, and only the nominal or updated admittance factors differ. Results coincide before switching. After changes, the largest VUF MAE decreases from 0.314 to 0.0164 p.p. on CIGRE and from 0.130 to 0.00700 p.p. on Kerber (Fig.~\ref{fig:operator_recomposition} and Table~\ref{tab:topology_mechanism}). Ratios of these maxima are 19.2 and 18.6, compared with 2.28 on EuLV. Remaining error includes approximation from holding positive-sequence voltage fixed. This diagnostic isolates the update contribution without identifying a single cause of differences among networks.
	
	\begin{figure}[!t]
		\centering
		\includegraphics[width=.92\columnwidth]{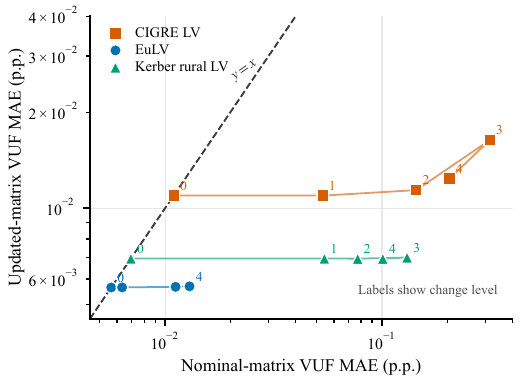}
		\caption{Paired admittance-update diagnostic with positive-sequence voltage fixed and equal iteration budgets. Each point represents one network and change level. Points below $y=x$ have smaller error with updated factors. Errors are in percentage points.}
		\label{fig:operator_recomposition}
	\end{figure}

	\begin{table}[!t]
		\centering
		\caption{Candidate-Transfer Acceptance and VUF Error after Topology Changes}
		\label{tab:topology_mechanism}
		\scriptsize
		\setlength{\tabcolsep}{2.2pt}
		\begin{tabular}{lrrrrr}
			\toprule
			Network & Min. $f_{\mathrm{tr}}$/\% & Max. $\gamma$ & Max. $E_{\mathrm{nom}}$/\% & Max. $E_{\mathrm{upd}}$/\% & Error reduction\\
			\midrule
			CIGRE LV & 43.1 & 2.41 & 0.314 & 0.0164 & 19.2\\
			EuLV & 100 & 0.669 & 0.0130 & 0.00569 & 2.28\\
			Kerber LV & 99.7 & 0.661 & 0.130 & 0.00700 & 18.6\\
			\bottomrule
		\end{tabular}
	\end{table}

	The candidate-transfer diagnostic checks whether a nominal-topology candidate can be enclosed under the changed coupled model. Let $f_{\mathrm{tr}}$ be the accepted fraction of retained records, and let $\gamma_\theta$ bound the fixed-point contraction factor. EuLV and Kerber have $\gamma_{\max}<0.670$ and rejection rates no greater than 0.300\% (Fig. 6). CIGRE exceeds $\gamma_{\max}=1$ after the first change. At change level three, 2775 of 4875 records are rejected: 1683 fail contraction and 1092 pass contraction but fail the voltage-margin check. These groups account for 34.5\% and 22.4\% of retained records, totaling 56.9\%. At level four, the maximum contraction bound decreases from 2.41 to 1.79 and rejection falls to 50.2\%. Acceptance therefore need not decline monotonically with the number of changes.

	\begin{figure}[!t]
		\centering
		\includegraphics[width=.92\columnwidth]{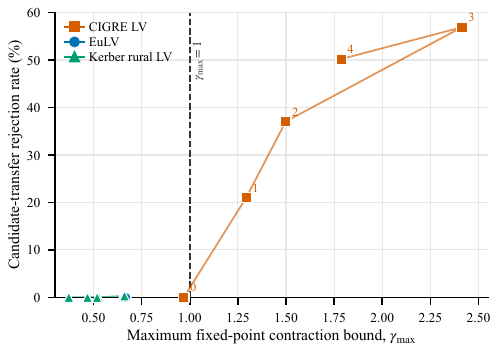}
		\caption{Candidate-transfer diagnostic with positive-sequence voltage fixed. Rejection is $100(1-f_{\mathrm{tr}})$ and is distinct from the Tier II fallback rate. The dashed line marks $\gamma_{\max}=1$.}
		\label{fig:network_validity_burden}
	\end{figure}

    All 15 rejected Kerber records pass contraction, illustrating the separate voltage-margin requirement. Let $r_0$ be the nominal-map residual and $d_\theta$ the change in the map output at the same candidate. For $\gamma_\theta<1$, the required transfer radius is $R_\theta=(r_0+d_\theta)/(1-\gamma_\theta)$. The median-radius rejected case has $r_0=1.37\times10^{-16}$ p.u., $d_\theta=0.0388$ p.u., and $\gamma_\theta=0.448$. The amplification factor is approximately 1.81, giving~\eqref{eq:migration_substitution}.

	\begin{equation}
		\begin{aligned}
			R_\theta&\simeq1.81(0.0388)\simeq0.0702\ \mathrm{p.u.},\\
			2R_\theta&\simeq0.140>0.132
			\simeq\min_p|\widehat v_p^{abc}|-0.800.
		\end{aligned}
		\label{eq:migration_substitution}
	\end{equation}
	The factor of two bounds the combined zero- and negative-sequence contribution to phase-voltage variation. The required radius exceeds the margin above the diagnostic floor of 0.800 p.u., leaving a recorded margin of $-8.77\times10^{-3}$ p.u. Thus, nominal-model convergence and a contraction bound below one do not suffice for transfer acceptance. Across all networks, 54,672 of 62,875 retained records pass the transfer check (87.0\%). Every accepted enclosure contains the changed equilibrium of the coupled model with positive-sequence voltage fixed.

    Across the 12 network/change combinations, TACQE's normalized quantile error has median 1.01, mean 1.15, and maximum 2.24 (Fig.~\ref{fig:stability}(b)). Baseline medians are 1.48 for DDA-L3PF, 2.88 for PCA-GP, and 3.57 for PCA-SVR. FP-$3\phi$ has a ratio of order $10^3$ because its unedited error is nearly zero. Table~\ref{tab:quantile_mae} supplies the absolute-error context. Under stress, TACQE's error increases by a factor of 2.24, compared with 7.02 for PCA-SVR (Fig.~\ref{fig:stability}(a)). Updating addresses topology-induced equation mismatch, while stress-specific recalibration addresses the changed operating distribution.

	\begin{figure}[!t]
		\centering
		\includegraphics[width=.485\columnwidth]{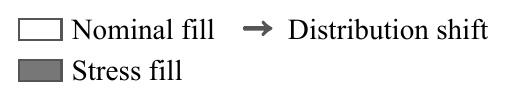}\hfill
		\includegraphics[width=.485\columnwidth]{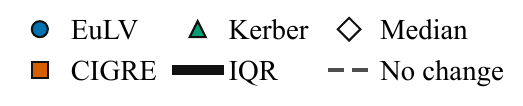}\\[-1mm]
		\includegraphics[width=.485\columnwidth]{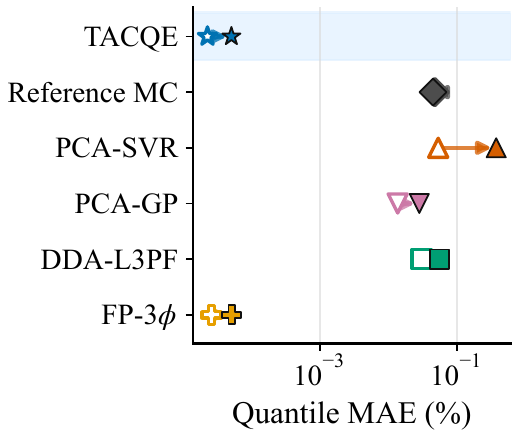}\hfill
		\includegraphics[width=.485\columnwidth]{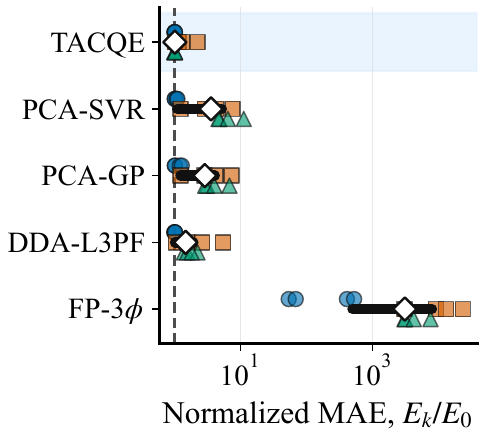}
		\caption{Quantile-error sensitivity to distribution and topology changes. (a) $E_{\mathrm{all}}$ under nominal and stress conditions. (b) Error after one to four topology changes, normalized by each method's unedited-topology error on the same network.}
		\label{fig:stability}
	\end{figure}
	
	\subsection{Reference-solver Evaluations and Local Update Cost}
	
	We assess reference-solver demand separately from local-update runtime. Reference counts include training, calibration, and deployment fallback, but exclude audit-only solutions. The runtime benchmark compares matrix and state consistency with full reconstruction across 36 topology variants.

    In the fixed-topology runs, TACQE uses 1000 calibration solutions and no fallback, compared with 1800 training and calibration solutions for each trained baseline (Fig.~\ref{fig:cost}(a)). This saves 44.4\% of reference evaluations while approaching FP-$3\phi$ point accuracy. Deployment requires $n_c+N_{\mathrm{II}}$ reference solutions, where $N_{\mathrm{II}}$ is the fallback count. With the present budgets, TACQE uses fewer solutions than a trained baseline only if $N_{\mathrm{II}}<800$, or less than 8.89\% of the 9000 evaluation scenarios. Note that the 87.0\% transfer-acceptance rate cannot be used to infer $N_{\mathrm{II}}$. That diagnostic fixes positive-sequence voltage, whereas full certification evaluates all unknown components under~\eqref{eq:exact}.
	
	\begin{figure}[!t]
		\centering
		\includegraphics[width=.485\columnwidth]{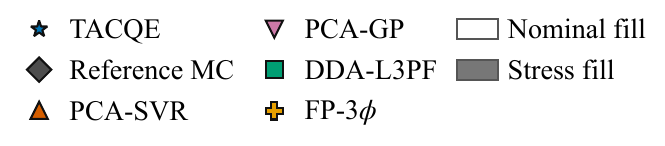}\hfill
		\includegraphics[width=.485\columnwidth]{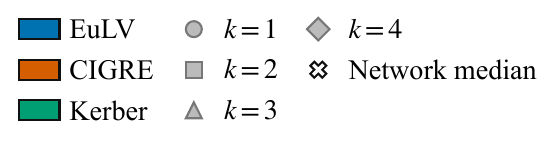}\\[-1mm]
		\includegraphics[width=.485\columnwidth]{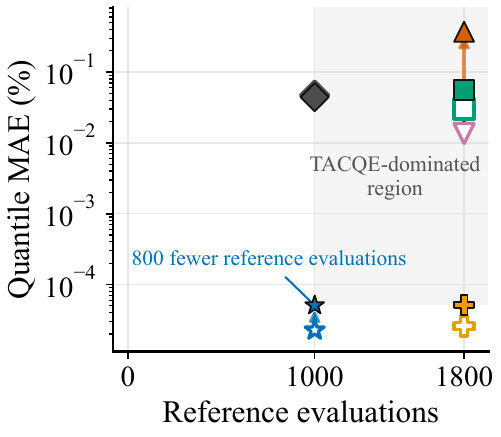}\hfill
		\includegraphics[width=.485\columnwidth]{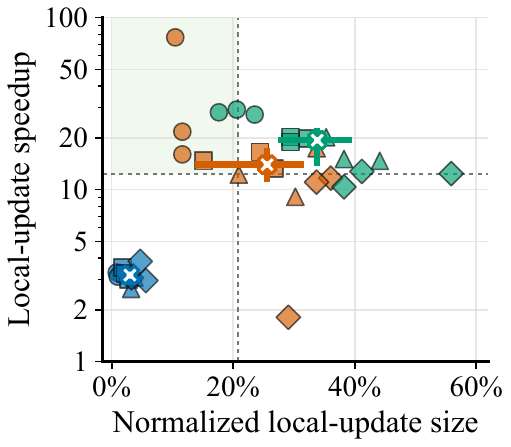}
		\caption{Accuracy and computational cost. (a) Reference-evaluation count versus quantile MAE. (b) Local-update speedup versus normalized update size across 36 topology variants. Crosses mark network medians. Timings in (b) cover the update routine only.}
		\label{fig:cost}
	\end{figure}

    Across 36 variants, the local update has median speedup 12.3 and range 1.81--76.8 (Fig.~\ref{fig:cost}(b)). Normalized update size is the number of edges in the local region divided by the combined edge counts before and after switching. For the first EuLV change, a 1.38\% update takes 77.6 ms versus 254 ms for full reconstruction, a speedup of approximately 3.27. Cached global-inverse operations remain costly even for small updates. EuLV has median update size 3.01\% and speedup 3.18, whereas CIGRE reaches speedup 14.0 at 25.6\%. Update size alone therefore does not explain cross-network acceleration. Compared with full reconstruction, all 36 variants have maximum absolute differences of $1.72\times10^{-9}$ in inverse entries and $1.37\times10^{-13}$ in state variables.

	The experiments distinguish numerical accuracy, model consistency, and statistical resolution. TACQE preserves fixed-topology point accuracy overall, although FP-$3\phi$ has smaller stress upper-tail error. Updating reduces error after topology changes. The results also show that numerical enclosure and reference expansion account for only about 0.231\% of fixed-topology probability-envelope width, and calibration uncertainty and endpoint sampling dominate. We also observe that reference demand depends on $n_c$ and fallback, while tail resolution depends mainly on $N$.

	\section{Conclusion}
	\label{sec:conclusion}
	
	TACQE integrates topology-updated power-flow equations, scenario-specific numerical state certification, and reference-based statistical calibration for maximum-VUF assessment under random power injections and topology variations. For each scenario, certification against the realized nonlinear power-flow equations provides a state enclosure and a corresponding interval for the maximum VUF. Independent reference-solver evaluations calibrate residual discrepancies between these intervals and the reference outputs. Calibrated endpoint exceedances, combined with a uniform DKW bound, yield simultaneous finite-sample bounds on reference upper-tail probabilities and the corresponding VUF quantiles over the selected threshold range.
    
    For fixed-topology cases, our simulation results show that TACQE achieves mean absolute errors of $2.26\times10^{-5}$ p.p. under nominal conditions and $5.06\times10^{-5}$ p.p. under stress conditions. Across topology-change experiments on three distribution networks, its quantile-function and upper-tail errors are smaller than those of four baselines by factors of 8.27--12.1 and 7.76--10.4, respectively. With equal reference-solver budgets, TACQE provides a median upper-tail probability-envelope width approximately 55\% smaller than that of Monte Carlo simulation. In addition, the local admittance-update routine achieves a median speedup of 12.3.
    
    We note that finite-sample uncertainty limit the resolution of extreme-quantile guarantees. With the current sampling budgets, the certified upper bound at the 99-th percentile remains infinite, despite accurate empirical point estimates. We therefore envision future work to examine sharper uniform tail bounds and recalibration procedures for changing operating distributions. We also suggest the investigation of the deployment of our approach, which will require treatment of network-parameter uncertainty, topology-identification uncertainty, and unresolved reference power-flow solutions.
	
	\bibliographystyle{IEEEtran}
	\bibliography{myref}

@article{Damianakis2025GridImpact,
  author  = {Damianakis, Nikolaos and Chandra Mouli, Gautham Ram and Bauer, Pavol},
  title   = {Grid Impact of Photovoltaics, Electric Vehicles and Heat Pumps on Distribution Grids---An Overview},
  journal = {Applied Energy},
  volume  = {380},
  pages   = {125000},
  year    = {2025},
  doi     = {10.1016/j.apenergy.2024.125000}
}

@article{Gruosso2024AdvancedVUF,
  author  = {Gruosso, Giambattista and Diaz Londono, Cesar and Daniel, Luca and Maffezzoni, Paolo},
  title   = {Advanced Probabilistic Load Flow Methodology for Voltage Unbalance Assessment in {PV} Penetrated Distribution Grids},
  journal = {International Journal of Electrical Power \& Energy Systems},
  volume  = {155},
  pages   = {109556},
  year    = {2024},
  doi     = {10.1016/j.ijepes.2023.109556}
}

@article{Tsobze2025VoltageImbalance,
  author  = {Tsobze, Saatong Kenfack and Mbasso, Wulfran Fendzi and Harrison, Ambe and Khishe, Mohammad and Dagal, Idriss and Jangir, Pradeep and Smerat, Aseel and Al-Gahtani, Saad F. and Elbarbary, Z. M. S.},
  title   = {A Decentralized Power Injection-Based Approach for Voltage Imbalance Mitigation in Three-Phase Distribution Networks},
  journal = {Scientific Reports},
  volume  = {15},
  number  = {1},
  pages   = {16162},
  year    = {2025},
  doi     = {10.1038/s41598-025-00328-4}
}

@inproceedings{Christianen2024ImportanceSampling,
  author    = {Christianen, Mark and Lam, Henry and Vlasiou, Maria and Zwart, Bert},
  title     = {Importance Sampling of Rare Events for Distribution Networks with Stochastic Loads},
  booktitle = {2024 Winter Simulation Conference (WSC)},
  pages     = {3590--3601},
  year      = {2024},
  doi       = {10.1109/WSC63780.2024.10838841}
}

@article{Martins2024KMedoids,
  author  = {Martins, Ant{\^o}nio Sobrinho Campolina and de Araujo, Leandro Ramos and Penido, D{\'e}bora Rosana Ribeiro},
  title   = {K-Medoids Clustering Applications for High-Dimensionality Multiphase Probabilistic Power Flow},
  journal = {International Journal of Electrical Power \& Energy Systems},
  volume  = {157},
  pages   = {109861},
  year    = {2024},
  doi     = {10.1016/j.ijepes.2024.109861}
}

@article{Liao2024ProbInterval,
  author  = {Liao, Xiaobing and Zhang, Yiming and Li, Zicheng and Wei, Hanqi and Ding, Hua},
  title   = {Probabilistic Interval Power Flow Calculation Method for Distribution Networks Considering the Correlation of Distributed Wind Power Output},
  journal = {International Journal of Electrical Power \& Energy Systems},
  volume  = {157},
  pages   = {109827},
  year    = {2024},
  doi     = {10.1016/j.ijepes.2024.109827}
}

@article{Lin2024LinDistS,
  author  = {Lin, Hanyang and Shen, Xinwei and Guo, Ye and Ding, Tao and Sun, Hongbin},
  title   = {A Linear {DistFlow} Model Considering Line Shunts for Fast Calculation and Voltage Control of Power Distribution Systems},
  journal = {Applied Energy},
  volume  = {357},
  pages   = {122467},
  year    = {2024},
  doi     = {10.1016/j.apenergy.2023.122467}
}

@article{Mahmoodi2024UnbalancedPF,
  author  = {Mahmoodi, Masoume and Noori R.A., S. Mahdi and Blackhall, Lachlan and Scott, Paul},
  title   = {A Comparison on Power Flow Models for Optimal Power Flow Studies in Integrated Medium-Low Voltage Unbalanced Distribution Systems},
  journal = {Sustainable Energy, Grids and Networks},
  volume  = {38},
  pages   = {101339},
  year    = {2024},
  doi     = {10.1016/j.segan.2024.101339}
}

@inproceedings{Chung2024TaylorRobustPF,
  author    = {Chung, Sungjoo and Zhang, Ying and Wang, Zhaoyu and Ding, Fei},
  title     = {Taylor-Expansion-Based Robust Power Flow in Unbalanced Distribution Systems---A Hybrid Data-Aided Method},
  booktitle = {2024 IEEE Power \& Energy Society General Meeting (PESGM)},
  pages     = {1--5},
  year      = {2024},
  doi       = {10.1109/PESGM51994.2024.10689162}
}

@article{Yang2024PGGNNPPF,
  author  = {Yang, Mei and Qiu, Gao and Liu, Tingjian and Liu, Junyong and Liu, Kai and Li, Yaping},
  title   = {Probabilistic Power Flow Based on Physics-Guided Graph Neural Networks},
  journal = {Electric Power Systems Research},
  volume  = {235},
  pages   = {110864},
  year    = {2024},
  doi     = {10.1016/j.epsr.2024.110864}
}

@article{Gao2024PhysicsGCN,
  author  = {Gao, Maosheng and Yu, Juan and Yang, Zhifang and Zhao, Junbo},
  title   = {Physics Embedded Graph Convolution Neural Network for Power Flow Calculation Considering Uncertain Injections and Topology},
  journal = {IEEE Transactions on Neural Networks and Learning Systems},
  volume  = {35},
  number  = {11},
  pages   = {15467--15478},
  year    = {2024},
  doi     = {10.1109/TNNLS.2023.3287028}
}

@article{Ghamizi2025PowerFlowMultiNet,
  author  = {Ghamizi, Salah and Cao, Jun and Ma, Aoxiang and Rodriguez, Pedro},
  title   = {{PowerFlowMultiNet}---Multigraph Neural Networks for Unbalanced Three-Phase Distribution Systems},
  journal = {IEEE Transactions on Power Systems},
  volume  = {40},
  number  = {1},
  pages   = {1148--1151},
  year    = {2025},
  doi     = {10.1109/TPWRS.2024.3465088}
}

@inproceedings{Glover2025MultiphaseGP,
  author    = {Glover, Daniel and Pareek, Parikshit and Deka, Deepjyoti and Dubey, Anamika},
  title     = {Power Flow Approximations for Multiphase Distribution Networks Using Gaussian Processes},
  booktitle = {2025 IEEE Power \& Energy Society General Meeting (PESGM)},
  pages     = {1--5},
  year      = {2025},
  doi       = {10.1109/PESGM52009.2025.11225305}
}

@article{Yang2024TopologyTransfer,
  author  = {Yang, Mei and Qiu, Gao and Liu, Junyong and Liu, Youbo and Liu, Tingjian and Tang, Zhiyuan and Ding, Lijie and Shui, Yue and Liu, Kai},
  title   = {Topology-Transferable Physics-Guided Graph Neural Network for Real-Time Optimal Power Flow},
  journal = {IEEE Transactions on Industrial Informatics},
  volume  = {20},
  number  = {9},
  pages   = {10857--10872},
  year    = {2024},
  doi     = {10.1109/TII.2024.3398058}
}

@article{Pareek2025PVE,
  author  = {Pareek, Parikshit and Misra, Sidhant and Deka, Deepjyoti},
  title   = {Data-Efficient Strategies for Probabilistic Voltage Envelopes under Network Contingencies},
  journal = {Sustainable Energy, Grids and Networks},
  volume  = {43},
  pages   = {101811},
  year    = {2025},
  doi     = {10.1016/j.segan.2025.101811}
}

@article{Tan2025RareRisk,
  author  = {Tan, Bendong and Zhao, Junbo and Chen, Yousu},
  title   = {Scalable Risk Assessment of Rare Events in Power Systems with Uncertain Wind Generation and Loads},
  journal = {IEEE Transactions on Power Systems},
  volume  = {40},
  number  = {2},
  pages   = {1374--1388},
  year    = {2025},
  doi     = {10.1109/TPWRS.2024.3435490}
}

@inproceedings{Niu2024VoltageRisk,
  author    = {Niu, Zhenyong and Wu, Renbo and Peng, Yiming},
  title     = {Voltage Risk Assessment of Low-Voltage Distribution Network Taking into Account the Correlation of {PV}},
  booktitle = {2024 IEEE 6th International Conference on Power, Intelligent Computing and Systems (ICPICS)},
  pages     = {363--367},
  year      = {2024},
  doi       = {10.1109/ICPICS62053.2024.10795891}
}

@article{Su2024DynamicBoundary,
  author  = {Su, Yu and Li, Dingrui and Wang, Fred and Olama, Mohammed M. and Ferrari, Maximiliano F. and Ollis, Ben and Liu, Yilu},
  title   = {Flexible Dynamic Boundary Microgrid Operation Considering Network and Load Unbalances},
  journal = {Applied Energy},
  volume  = {371},
  pages   = {123633},
  year    = {2024},
  doi     = {10.1016/j.apenergy.2024.123633}
}

@article{Gao2024BayesianPOPF,
  author  = {Gao, Fang and Xu, Zidong and Yin, Linfei},
  title   = {Bayesian Deep Neural Networks for Spatio-Temporal Probabilistic Optimal Power Flow with Multi-Source Renewable Energy},
  journal = {Applied Energy},
  volume  = {353},
  pages   = {122106},
  year    = {2024},
  doi     = {10.1016/j.apenergy.2023.122106}
}

@article{Wilks1941Tolerance,
  author  = {Wilks, Samuel S.},
  title   = {Determination of Sample Sizes for Setting Tolerance Limits},
  journal = {The Annals of Mathematical Statistics},
  volume  = {12},
  number  = {1},
  pages   = {91--96},
  year    = {1941},
  doi     = {10.1214/aoms/1177731788}
}

@article{Dvoretzky1956DKW,
  author  = {Dvoretzky, Aryeh and Kiefer, Jack and Wolfowitz, Jacob},
  title   = {Asymptotic Minimax Character of the Sample Distribution Function and of the Classical Multinomial Estimator},
  journal = {The Annals of Mathematical Statistics},
  volume  = {27},
  number  = {3},
  pages   = {642--669},
  year    = {1956},
  doi     = {10.1214/aoms/1177728174}
}

@article{Massart1990DKW,
  author  = {Massart, Pascal},
  title   = {The Tight Constant in the {Dvoretzky--Kiefer--Wolfowitz} Inequality},
  journal = {The Annals of Probability},
  volume  = {18},
  number  = {3},
  pages   = {1269--1283},
  year    = {1990},
  doi     = {10.1214/aop/1176990746}
}

@misc{IEEEPES2015EuLV,
  author       = {{IEEE PES AMPS DSAS Test Feeder Working Group}},
  title        = {European Low Voltage Test Feeder},
  year         = {2015},
  howpublished = {IEEE Power and Energy Society Test Feeder Repository}
}

@techreport{CIGRE2014Benchmark,
  author      = {Strunz, Kai and Abbasi, Ehsan and Fletcher, Robert and Hatziargyriou, Nikos and Iravani, Reza and Joos, Geza and Le-Huy, Hoang},
  title       = {Benchmark Systems for Network Integration of Renewable and Distributed Energy Resources},
  institution = {CIGRE},
  number      = {Technical Brochure 575},
  year        = {2014}
}

@phdthesis{Kerber2011LV,
  author = {Kerber, Georg},
  title  = {Aufnahmef{\"a}higkeit von Niederspannungsverteilnetzen f{\"u}r die Einspeisung aus Photovoltaikkleinanlagen},
  school = {Technical University of Munich},
  year   = {2011}
}

@inproceedings{Drucker1997SVR,
  author    = {Drucker, Harris and Burges, Christopher J. C. and Kaufman, Linda and Smola, Alex and Vapnik, Vladimir},
  title     = {Support Vector Regression Machines},
  booktitle = {Advances in Neural Information Processing Systems 9},
  pages     = {155--161},
  year      = {1997}
}

@book{Jolliffe2002PCA,
  author    = {Jolliffe, Ian T.},
  title     = {Principal Component Analysis},
  edition   = {2},
  series    = {Springer Series in Statistics},
  publisher = {Springer},
  address   = {New York, NY, USA},
  year      = {2002},
  doi       = {10.1007/b98835}
}

@book{Rasmussen2006GPML,
  author    = {Rasmussen, Carl Edward and Williams, Christopher K. I.},
  title     = {Gaussian Processes for Machine Learning},
  publisher = {MIT Press},
  address   = {Cambridge, MA, USA},
  year      = {2006},
  isbn      = {9780262182539}
}

@article{Liu2022DDAL3PF,
  author  = {Liu, Yitong and Li, Zhengshuo and Zhou, Yu},
  title   = {Data-Driven-Aided Linear Three-Phase Power Flow Model for Distribution Power Systems},
  journal = {IEEE Transactions on Power Systems},
  volume  = {37},
  number  = {4},
  pages   = {2783--2795},
  year    = {2022},
  doi     = {10.1109/TPWRS.2021.3130301}
}

@article{Nguyen2024FP3Phi,
  author  = {Nguyen, Hoang Tien and Kim, Young-Jin and Choi, Dae-Hyun},
  title   = {Sample-Efficient Learning for a Surrogate Model of Three-Phase Distribution System},
  journal = {IEEE Transactions on Power Systems},
  volume  = {39},
  number  = {1},
  pages   = {2361--2364},
  year    = {2024}
}
	
\end{document}